\documentclass[american]{article}
\usepackage[T1]{fontenc}
\usepackage[latin9]{inputenc}
\usepackage{amsmath}
\usepackage{amssymb}
\usepackage{amsthm}
\usepackage{epsfig}
\usepackage{amssymb}
\usepackage{setspace}
\usepackage{fullpage}

\usepackage{babel}

\newtheorem{Theorem}{Theorem}
\newtheorem{Lemma}{Lemma}
\newtheorem{proposition}{Proposition}

\newtheorem{defn}{Definition}

\begin{document}

\title{Existence of Stable Exclusive Bilateral Exchanges in Networks}

\author{Ankur Mani, Asuman Ozdaglar, Alex (Sandy) Pentland\\
amani@mit.edu, asuman@mit.edu, pentland@mit.edu}

\date{}

\maketitle

\abstract

In this paper we show that when individuals in a bipartite network exclusively choose partners and exchange valued goods with their partners, then there exists a set of exchanges that are pair-wise stable. Pair-wise stability implies that no individual breaks her partnership and no two neighbors in the network can form a new partnership while breaking other partnerships if any so that at least one of them improves her payoff and the other one does at least as good. We consider a general class of continuous, strictly convex and strongly monotone preferences over bundles of goods for individuals. Thus, this work extends the general equilibrium framework from markets to networks with exclusive exchanges. We present the complete existence proof using the existence of a generalized stable matching in \cite{Generalized-Stable-Matching}. The existence proof can be extended to problems in social games as in \cite{Matching-Equilibrium} and \cite{Social-Games}.

\large
\doublespacing


\section{Introduction}

In this paper we show that when individuals in a bipartite network exclusively choose partners and exchange valued goods with their partners, then there exists a set of exchanges that are pair-wise stable. Pair-wise stability implies that no individual breaks her partnership and no two neighbors in the network can form a new partnership while breaking other partnerships if any so that at least one of them improves her payoff and the other one does at least as good. We consider a general class of continuous, strictly convex and strongly monotone preferences over bundles of goods for individuals. Thus, this work extends the general equilibrium framework from markets to networks with exclusive exchanges. However, unlike the general equilibrium, the strategy set of individuals is not a continuous demand but is hybrid, since along with demand, the individuals also choose the partners. A simplified version of this problem has been studied in the context of the assignment game \cite{Shapley-Shubik-1972} where the authors only consider indivisible goods. The existence of a stable set of exchanges is shown through the existence of a feasible solution for the associated relaxed linear program. However, the same methodology cannot be extended to the divisible goods case unless restrictive assumptions are made about the preferences. A similar problem has been studied in \cite{Matching-Equilibrium}. The authors study a similar concept called Matching Equilibrium in Bipartite social games and show its existence by the convergence of the Gale-Shapley deferred acceptance algorithm \cite{Gale-Shapley}. However, even in their case the Gale-Shapley deferred acceptance algorithm cannot be applied to the case where the set of Nash Equilibria for any pair of individuals is uncountable. We present the complete existence proof using the existence of a generalized stable matching in \cite{Generalized-Stable-Matching}. The existence proof can be extended to problems in social games as in \cite{Matching-Equilibrium} and \cite{Social-Games}.

The organization of the rest of the paper is as such. In section 2, we present the model of bilateral exchanges in networks and characterize the properties of such exchanges. In section 3, we introduce the network exchange game. In section 4 we present the solution concept of pair-wise stability and in section 5 we prove the existence of a pair-wise stable strategy profile in the network exchange game.


\section{Model}

In this section we introduce the model of network exchange. We first introduce the network, bilateral exchanges and pareto-efficient exchanges and payoffs and then characterize the sets of exchanges and payoffs. and the relations between them.


\subsection{Network}

A {\bf social network} is a weighted undirected graph $S=(N,E,W)$, where $N$ is the set of actors, $E$ is the set of links between actors ($E\subseteq N\times N$) and $W:E\rightarrow\mathbf{R}_{+}$ is a function over the links and represents the capacity of the link. The set of neighbors of a node $i \in N$ is $Nbr\left(i\right) = \{j \in N : \left(i,j\right) \in E\}$.

We define $M=\left\{ X,Y\right\} $ as the set of types of items of value that can be exchanged between the actors. Each actor $i\in N$ has exactly one type of item $M_{i}\in M$. The amount of item $M_{i}$ with actor $i$ is $|M_{i}|$. Each actor $i\in N$ has {\bf continuous, strictly convex and strongly monotone preferences} $\preceq_i$ over bundles of items. Alternatively, each actor $i$ has a {\bf continuous, strictly quasi-concave and strongly monotone utility function} $\pi_{i}:\mathbf{R}_{+}^{2} \rightarrow \mathbf{R}_{+}$.

The set of actors can be divided into two disjoint sets based upon the item they have. Let $A=\left\{ i\in N:M_{i}=X\right\} $ and $B=\left\{ j\in N:M_{j}=Y\right\} $. Then an actor in $A$ can only perform a rationally feasible exchange with another actor in $B$. Thus without loss of generality, we restrict our attention to bipartite social networks in which $E\in A\times B$. We will call the actors in $A$, the buyers and the actors in $B$, the sellers.


\subsection{Bilateral Exchanges}

An {\bf exchange} between an actor $i$ and $j$ with $\left(i,j\right) \in E$, is a tuple $(m_i,m_j) \in [0,|M_i|] \times [0,|M_j|]$ and involves $i$ giving amount $m_{i}$ units of item $M_{i}$ to actor $j$ and $j$ giving some amount $m_{j}$ units of item $M_{j}$ to actor $i$.

Each link $(i,j)\in E$ is an {\bf exchange opportunity} of capacity $W\left(i,j\right)$ between $i$ and $j$. The capacity of an exchange opportunity $\left(i,j\right)$ is the maximum allowed amount of items that can be exchanged between the actors $i$ and $j$, i.e.- for any exchange $\left(m_i,m_j\right)$ between $i$, and $j$, $m_i+m_j \leq W\left(i,j\right)$. The set of possible exchanges $EX\left(i,j\right)$ between neighboring actors $i,j$ is a polytope specified by the following inequalities.
\begin{align}
&0 \leq m_i \leq |M_i|\\
&0 \leq m_j \leq |M_j|\\
&m_i+m_j \leq W\left(i,j\right)
\end{align}
The {\bf payoff function} for actor $i$ in an exchange opportunity $\left(i,j\right)$ is $V_{ij} ~:~ EX\left(i,j\right) \rightarrow \mathbf{R}$ and the {\bf payoff function} for actor $j$ is $V_{ji} ~:~ EX\left(i,j\right) \rightarrow \mathbf{R}$. In an exchange $\left(m_i,m_j\right)$ between $i$ and $j$, the {\bf payoff} of actor $i$ is $V_{ij}\left(m_{i},m_{j}\right)=\pi_{i}\left(|M_{i}|-m_{i},m_{j}\right)-\pi_{i}\left(|M_{i}|,0\right)$ and the {\bf payoff} of actor $j$ is $V_{ji}\left(m_{i},m_{j}\right)=\pi_{j}\left(m_{i},|M_j|-m_{j}\right)-\pi_{j}\left(0, |M_{j}|\right)$. From the assumptions on $\pi_i$ and $\pi_j$, the functions $V_{ij}$ and $V_{ji}$ are continuous, strictly quasi-concave and strongly monotone. $V_{ij}$ is increasing in $m_j$ and decreasing in $m_i$ and $V_{ji}$ is increasing in $m_i$ and decreasing in $m_j$.
The set $RX\left(i,j\right) \subset EX\left(i,j\right)$ of  {\bf rationally feasible} exchanges between actors $i$ and $j$ is the convex and compact set specified by the following inequalities.
\begin{align}
&0 \leq m_i \leq |M_i|\\
&0 \leq m_j \leq |M_j|\\
&m_i+m_j \leq W\left(i,j\right)\\
&V_{ij}\left(m_{i},m_{j}\right)\geq0\\
&V_{ji}\left(m_{j},m_{i}\right)\geq0
\end{align}
The convexity comes from the fact that $V_{ij}$ and $V_{ji}$ are strictly quasi-concave and $EX\left(i,j\right)$ is a polytope.

A {\bf payoff vector function} $\mathbf{V}_{i,j} ~:~ EX\left(i,j\right) \rightarrow \mathbf{R}^2$ maps the exchanges to the payoffs of the actors.
The {\bf payoff vector} for exchange $\left(m_i,m_j\right)$ in the exchange opportunity $\left(i,j\right)$ is a tuple $\mathbf{V}_{i,j}\left(m_{i},m_{j}\right) = \left(V_{ij}\left(m_{i},m_{j}\right),V_{ji}\left(m_{j},m_{i}\right)\right)$.
From the continuity of $V_{ij}$ and $V_{ji}$, $\mathbf{V}_{i,j}$ is continuous.

The set of {\bf possible payoff vectors} $XV\left(i,j\right)$ is the image of $EX\left(i,j\right)$ under the function $\mathbf{V}_{i,j}$.
The set of {\bf rationally feasible payoff vectors} $RV\left(i,j\right)$ is the image of $RX\left(i,j\right)$ under the function $\mathbf{V}_{i,j}$.

The relationship between the sets of exchanges and the set of payoff vectors is captured in figure \ref{fig:Exchange-Payoff-Sets}.

\begin{figure} \label{fig:Exchange-Payoff-Sets}
\centering
\includegraphics[scale = 0.3]{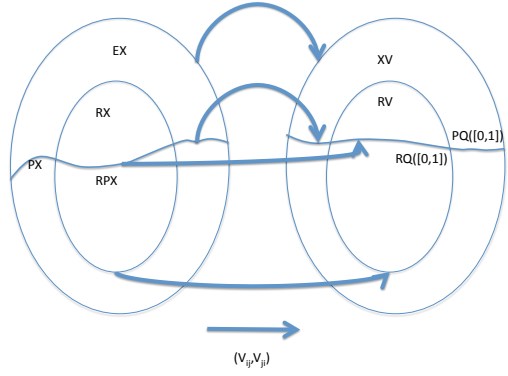}
\caption{The possible exchange set ($EX(i,j)$) includes rationally feasible exchange set ($RX(i,j)$) and the pareto-efficient exchange set ($PX(i,j)$) includes the rationally feasible pareto-efficient exchange set ($RPX(i,j)$). The payoff-vector function $\mathbf{V}_{i,j}$ maps these exchange sets to the set of payoff-vectors. The possible exchange set $EX(i,j)$ is mapped onto the set of possible payoff vectors $XV(i,j)$. The rationally feasible exchange set $RX(i,j)$ is mapped onto the set of rationally feasible payoff vectors $RV(i,j)$. The pareto-efficient exchange set $PX(i,j)$ is mapped onto the set of pareto-efficient payoff vectors ${PQ}^(i,j)([0,1])$. The rationally feasible pareto-efficient exchange set $RPX(i,j)$ is mapped onto the set of rationally feasible pareto-efficient payoff vectors ${RQ}^(i,j)([0,1])$. $XV(i,j)$ includes $RV(i,j)$ and ${PQ}^(i,j)([0,1])$ includes ${RQ}^(i,j)([0,1])$. $\mathbf{PQ}^(i,j)$ is a path in $XV(i,j)$ and ${RQ}^(i,j)$ is a path in $RV(i,j)$.}
\end{figure}

We now characterize the set of possible payoff vectors and the set of rationally feasible payoff vectors in the following proposition.

\begin{proposition}
The set of possible payoff vectors $XV\left(i,j\right)$ and the set of rationally feasible payoff vectors $RV\left(i,j\right)$ in the exchange opportunity $\left(i,j\right)$ are connected and compact.
\end{proposition}

\begin{proof}

We know that $EX\left(i,j\right)$ and $RX\left(i,j\right)$ are convex and hence connected. Since $\mathbf{V}_{i,j}$ is continuous, therefore by the intermediate value theorem \cite{Real-Analysis} for connected spaces, the images of $EX\left(i,j\right)$ and $RX\left(i,j\right)$ under $\mathbf{V}_{i,j}$ are connected. Hence, $XV\left(i,j\right)$ and $RV\left(i,j\right)$ are connected. Also, since $EX\left(i,j\right)$ and $RX\left(i,j\right)$ are compact, their images under the continuous function $\mathbf{V}_{i,j}$ are compact. Hence, $XV\left(i,j\right)$ and $RV\left(i,j\right)$ are compact.

\end{proof}

From the strong monotonicity property of $V_{ij}$ and $V_{ji}$,the maximum payoff for $i$ over all exchanges in $EX\left(i,j\right)$ is $V^*_i\left(EX\left(i,j\right)\right) = V_{ij}\left(0,\min\{|M_j|, W(i,j)\}\right)$ and the maximum payoff for $j$ over all exchanges in $EX\left(i,j\right)$ is $V^*_j\left(EX\left(i,j\right)\right) = V_{ji}\left(\min\{|M_i|, W(i,j)\},0\right)$. The minimum payoff for $i$ over all exchanges in $EX\left(i,j\right)$ is $V^\circ_i\left(EX\left(i,j\right)\right) = V_{ij}\left(\min\{|M_i|, W(i,j)\},0\right)$ and the maximum payoff for $j$ over all exchanges in $EX\left(i,j\right)$ is\\
$V^\circ_j\left(EX\left(i,j\right)\right) = V_{ji}\left(0,\min\{|M_j|, W(i,j)\}\right)$.

We will represent the maximum rationally feasible payoff for $i$ over $RX\left(i,j\right)$ as $V^*_i\left(RX\left(i,j\right)\right)$ and the maximum rationally feasible payoff for $j$ over $RX\left(i,j\right)$ is $V^*_j\left(RX\left(i,j\right)\right)$. The minimum rationally feasible payoff for $i$ over $RX\left(i,j\right)$ is $V^\circ_i\left(RX\left(i,j\right)\right)$ and the maximum rationally feasible payoff for $j$ over $RX\left(i,j\right)$ is $V^\circ_j\left(RX\left(i,j\right)\right)$.

From the intermediate value theorem, we know that for any\\
$x \in [V^\circ_i\left(EX\left(i,j\right)\right), V^*_i\left(EX\left(i,j\right)\right)]$, there exists an exchange $\left(m_i,m_j\right) \in EX(i,j)$ for which the payoff for $i$ is $V_{ij}\left(m_i,m_j\right) = x$. Similarly, for any $x \in [V^\circ_j\left(EX\left(i,j\right)\right), V^*_j\left(EX\left(i,j\right)\right)]$, there exists an exchange $\left(m_i,m_j\right) \in EX(i,j)$ for which the payoff for $j$ is $V_{ji}\left(m_i,m_j\right) = x$. Therefore the set of possible payoffs for $i$ is $[V^\circ_i\left(EX\left(i,j\right)\right), V^*_i\left(EX\left(i,j\right)\right)]$ and the set of possible payoffs for $j$ is $[V^\circ_j\left(EX\left(i,j\right)\right), V^*_j\left(EX\left(i,j\right)\right)]$.


\subsection{Pareto-efficient exchanges}

The set of {\bf pareto-efficient exchanges} or {\bf pareto-set} between two actors $i,j$ is
\begin{align}
&PX\left(i,j\right) = \{\left(m_i,m_j\right) \in EX\left(i,j\right) ~:~ \nexists \left(m^{'}_i,m^{'}_j\right) \in EX\left(i,j\right) \mbox{, with } \mathbf{V}_{i,j}\left(m^{'}_i,m^{'}_j\right)>\mathbf{V}_{i,j}\left(m^{'}_i,m^{'}_j\right)\}
\end{align}
The set of pareto-efficient payoff vectors is the image of the pareto-set under the payoff-vector function $\mathbf{V}_{i,j}$.

For any $x \in [V^\circ_i\left(EX\left(i,j\right)\right), V^*_i\left(EX\left(i,j\right)\right)]$ and $y \in [V^\circ_j\left(EX\left(i,j\right)\right), V^*_j\left(EX\left(i,j\right)\right)]$, we define the {\bf upper-level sets} $U_i(x) = \{\left(m_{i},m_{j}\right) \in EX\left(i,j\right) ~:~ V_{ij}\left(m_{i},m_{j}\right) \geq x\}$ and $U_j(y) = \{\left(m_{i},m_{j}\right) ~:~ V_{ji}\left(m_{i},m_{j}\right) \geq y\}$. 

Given, $q_{i} \in [V^\circ_i\left(EX\left(i,j\right)\right), V^*_i\left(EX\left(i,j\right)\right)]$ a pareto-efficient exchange $\left(m^{p_i}_i\left(q_i\right),m^{p_i}_j\left(q_i\right)\right)$ between actors $i$ and $j$ that gives a payoff of at least $q_i$ to actor $i$ is characterized by
\begin{align} \label{eqn:char-pareto-efficient}
&\left(m^{p_i}_{i}\left(q_{i}\right),m^{p_i}_{j}\left(q_{i}\right)\right) \in \underset{\left(m_{i},m_{j}\right) \in U_i\left(q_i\right)}{\arg\max}V_{ji}\left(m_{j},m_{i}\right)
\end{align}
Clearly, there exists at least one such exchange because $q_{i} \in [V^\circ_i\left(EX\left(i,j\right)\right), V^*_i\left(EX\left(i,j\right)\right)]$.

\begin{proposition} \label{prop:uniqueness-monotonicity-pareto}

The characteristics of the pareto-efficient exchanges that satisfy relation \ref{eqn:char-pareto-efficient} for some $q_{i} \in [V^\circ_i\left(EX\left(i,j\right)\right), V^*_i\left(EX\left(i,j\right)\right)]$ are:

\begin{enumerate}

\item For each $q_{i} \in [V^\circ_i\left(EX\left(i,j\right)\right), V^*_i\left(EX\left(i,j\right)\right)]$, there is a unique pareto-efficient exchange $\left(m^{p_i}_{i}\left(q_{i}\right),m^{p_i}_{j}\left(q_{i}\right)\right)$ that satisfies relation \ref{eqn:char-pareto-efficient}.

\item The payoff $V_{ij}\left(m^{p_i}_{i}\left(q_{i}\right),m^{p_i}_{j}\left(q_{i}\right)\right) = q_{i}$.

\item The payoff $V_{ji}\left(m^{p_i}_{i}\left(q_{i}\right),m^{p_i}_{j}\left(q_{i}\right)\right)$ is strictly decreasing in $q_i$.

\end{enumerate}

\end{proposition}

\begin{proof}

Firstly, we notice that since the payoff functions $V_{ij}$ and $V_{ji}$ are strictly quasi-concave and continuous, therefore for any $$x \in [V^\circ_i\left(EX\left(i,j\right)\right), V^*_i\left(EX\left(i,j\right)\right)]$$ and $$y \in [V^\circ_j\left(EX\left(i,j\right)\right), V^*_j\left(EX\left(i,j\right)\right)]$$ the upper-level sets $U_i\left(x\right)$ and $U_j\left(y\right)$ are strictly convex and compact.
For claim 1, assume there are at least two pareto-efficient exchanges $\left(m^{'}_{i},m^{'}_{j}\right)$ and $\left(m^{''}_{i},m^{''}_{j}\right)$ that satisfy \ref{eqn:char-pareto-efficient}. Then $V_{ji}\left(m^{'}_{i},m^{'}_{j}\right) = V_{ji}\left(m^{''}_{i},m^{''}_{j}\right)$ and $V_{ij}\left(m^{'}_{i},m^{'}_{j}\right) \geq q_i$, $V_{ji}\left(m^{''}_{i},m^{''}_{j}\right) \geq q_i$. Define $\left(m^{'''}_{i},m^{'''}_{j}\right) = 0.5\left(m^{'}_{i},m^{'}_{j}\right) + 0.5\left(m^{''}_{i},m^{''}_{j}\right)$. Then $\left(m^{'''}_{i},m^{'''}_{j}\right)$ belongs in both the upper-level sets and:
\begin{align*}
&V_{ji}\left(m^{'''}_{i},m^{'''}_{j}\right) > V_{ji}\left(m^{'}_{i},m^{'}_{j}\right)\\
&V_{ji}\left(m^{'''}_{i},m^{'''}_{j}\right) > V_{ji}\left(m^{''}_{i},m^{''}_{j}\right)\\
&V_{ij}\left(m^{'''}_{i},m^{'''}_{j}\right) > q_i
\end{align*}
This contradicts the assumption that both $\left(m^{'}_{i},m^{'}_{j}\right)$ and $\left(m^{''}_{i},m^{''}_{j}\right)$ satisfy \ref{eqn:char-pareto-efficient}. Therefore claim 1 holds.

For claim 2 assume $V_{ij}\left(m^{p_i}_{i}\left(q_{i}\right),m^{p_i}_{j}\left(q_{i}\right)\right) > q_{i}$. There are three possible cases as discussed next.

Case 1: If $m^{p_i}_j\left(q_i\right) > 0$. Then consider the exchange, $\left(m^{p_i}_{i}\left(q_{i}\right),0\right)$. From the intermediate value theorem, there exists $\lambda \in (0,1)$, such that
$$q_i < V_{ij}\left(m^{p_i}_{i}\left(q_{i}\right),\lambda m^{p_i}_{j}\left(q_{i}\right)\right) < V_{ij}\left(m^{p_i}_{i}\left(q_{i}\right),m^{p_i}_{j}\left(q_{i}\right)\right).$$
Therefore $\left(m^{p_i}_{i}\left(q_{i}\right),\lambda m^{p_i}_{j}\left(q_{i}\right)\right) \in U_i\left(q_i\right)$. From the strong monotonicity of $V_{ji}$,\\
$$V_{ji}\left(m^{p_i}_{i}\left(q_{i}\right),\lambda m^{p_i}_{j}\left(q_{i}\right)\right) > V_{ji}\left(m^{p_i}_{i}\left(q_{i}\right),m^{p_i}_{j}\left(q_{i}\right)\right).$$
This contradicts the assumption that $\left(m^{p_i}_{i}\left(q_{i}\right),m^{p_i}_{j}\left(q_{i}\right)\right)$ satisfies \ref{eqn:char-pareto-efficient}.

Case 2: If $m^{p_i}_j\left(q_i\right) = 0$ but $m^{p_i}_i\left(q_i\right) < \min\{|M_i|, W\left(i,j\right)\}$. Then consider the exchange, $\left(m_i,m_j\right) = \left(\min\{|M_i|, W\left(i,j\right)\},m^{p_i}_j\left(q_i\right)\right)$. From the intermediate value theorem, there exists $\lambda \in (0,1)$, such that
$$q_i < V_{ij}\left(\left(1-\lambda\right) m^{p_i}_{i}\left(q_{i}\right) + \lambda m_i, m^{p_i}_{j}\left(q_{i}\right)\right) < V_{ij}\left(m^{p_i}_{i}\left(q_{i}\right),m^{p_i}_{j}\left(q_{i}\right)\right).$$
Therefore $\left(\left(1-\lambda\right) m^{p_i}_{i}\left(q_{i}\right) + \lambda m_i, m^{p_i}_{j}\left(q_{i}\right)\right) \in U_i\left(q_i\right)$. From the strong monotonicity of $V_{ji}$,
$$V_{ji}\left(\left(1-\lambda\right) m^{p_i}_{i}\left(q_{i}\right) + \lambda m_i, m^{p_i}_{j}\left(q_{i}\right)\right) > V_{ji}\left(m^{p_i}_{i}\left(q_{i}\right),m^{p_i}_{j}\left(q_{i}\right)\right).$$
This contradicts the assumption that $\left(m^{p_i}_{i}\left(q_{i}\right),m^{p_i}_{j}\left(q_{i}\right)\right)$ satisfies \ref{eqn:char-pareto-efficient}. In either of the two cases, by contradiction, $V_{ij}\left(m^{p_i}_{i}\left(q_{i}\right),m^{p_i}_{j}\left(q_{i}\right)\right) = q_{i}$.

Case 3: If $m^{p_i}_j\left(q_i\right) = 0$ but $m^{p_i}_i\left(q_i\right) = \min\{|M_i|, W\left(i,j\right)\}$, then $V_{ij}\left(m^{p_i}_{i}\left(q_{i}\right),m^{p_i}_{j}\left(q_{i}\right)\right) = V^\circ_i\left(EX\left(i,j\right)\right) \leq q_{i}$ which contradicts the assumption. Therefore by contradiction,
$$V_{ij}\left(m^{p_i}_{i}\left(q_{i}\right),m^{p_i}_{j}\left(q_{i}\right)\right) = q_{i}.$$

For claim 3 pick any $q^{'}_i,q^{''}_i \in [V^\circ_i\left(EX\left(i,j\right)\right), V^*_i\left(EX\left(i,j\right)\right)]$ with $q^{'}_i>q^{''}_i$. Then since $\left(m^{p_i}_{i}\left(q^{''}_{i}\right),m^{p_i}_{j}\left(q^{''}_{i}\right)\right)$ is pareto-efficient, therefore
$$V_{ji}\left(m^{p_i}_{i}\left(q^{''}_{i}\right),m^{p_i}_{j}\left(q^{''}_{i}\right)\right) > V_{ji}\left(m^{p_i}_{i}\left(q^{'}_{i}\right),m^{p_i}_{j}\left(q^{'}_{i}\right)\right).$$
Thus claim 3 holds.

\end{proof}

For a pair of neighboring actors $\left(i,j\right)$ in the network, we define the functions
\begin{align}
&V^p_{ji} ~:~ [V^\circ_i\left(EX\left(i,j\right)\right), V^*_i\left(EX\left(i,j\right)\right)] \rightarrow [V^\circ_j\left(EX\left(i,j\right)\right), V^*_j\left(EX\left(i,j\right)\right)]\\
&V^p_{ij} ~:~ [V^\circ_j\left(EX\left(i,j\right)\right), V^*_j\left(EX\left(i,j\right)\right)] \rightarrow [V^\circ_i\left(EX\left(i,j\right)\right), V^*_i\left(EX\left(i,j\right)\right)]\\
&\mbox{where, } V^p_{ji}\left(x\right) = V_{ji}\left(m^{p_i}_{i}\left(x\right),m^{p_i}_{j}\left(x\right)\right) \mbox{ and } V^p_{ij}\left(x\right) = V_{ij}\left(m^{p_j}_{i}\left(x\right),m^{p_j}_{j}\left(x\right)\right).
\end{align}

\begin{proposition}

The functions $V^p_{ji}$ and $V^p_{ij}$ are inverse of each other, strictly decreasing and continuous.

\end{proposition}

\begin{proof}

Firstly using proposition \ref{prop:uniqueness-monotonicity-pareto}, we notice that $V^p_{ji}\left(x\right)$ is uniquely defined for each $x$ in its domain and $V^p_{ji}$ is strictly decreasing. Similarly, $V^p_{ij}\left(x\right)$ is uniquely defined for each $x$ in its domain and $V^p_{ij}$ is strictly decreasing.

Pick $x \in [V^\circ_i\left(EX\left(i,j\right)\right), V^*_i\left(EX\left(i,j\right)\right)]$ and $y = V^p_{ji}\left(x\right)$. Since $\left(x,y\right)$ is the unique pareto-efficient payoff with $j$'s payoff as $y$, therefore $V^p_{ij}\left(y\right) = x = {V^{p^{-1}}_{ji}\left(V^p_{ji}\left(x\right)\right)}$. Also, since $V^p_{ij}$ is uniquely defined over its domain, therefore the range of $V^p_{ji}$ is  $[V^\circ_j\left(EX\left(i,j\right)\right), V^*_j\left(EX\left(i,j\right)\right)]$. Similarly, the range of $V^p_{ij}$ is  $[V^\circ_i\left(EX\left(i,j\right)\right), V^*_i\left(EX\left(i,j\right)\right)]$. Therefore, $V^p_{ji}$ and $V^p_{ij}$ are inverse of each other.

For the continuity, consider any interval $\left[V^\circ_j\left(EX\left(i,j\right)\right),y\right] \subseteq [V^\circ_j\left(EX\left(i,j\right)\right), V^*_j\left(EX\left(i,j\right)\right)]$.
Define $x = V^{p^{-1}}_{ji}\left(y\right) = V^p_{ij}\left(y\right)$.
Then since $V^p_{ji}$ and $V^p_{ij}$ are strictly decreasing,
\begin{align*}
&\forall z \in \left[V^\circ_j\left(EX\left(i,j\right)\right),y\right], V^{p^{-1}}_{ji}\left(z\right) = V^{p}_{ij}\left(z\right) \in  [x, V^*_i\left(EX\left(i,j\right)\right)]\\
&\mbox{and } \forall z \in \left[x, V^*_i\left(EX\left(i,j\right)\right)\right], V^{p}_{ji}\left(z\right) \in  [V^\circ_j\left(EX\left(i,j\right)\right),y]
\end{align*}
Since, $y$ were arbitrarily picked, therefore the pre-image of all closed left intervals in\\
$[V^\circ_j\left(EX\left(i,j\right)\right), V^*_j\left(EX\left(i,j\right)\right)]$ are closed. Therefore, $V^p_{ji}$ is continuous. By similar reasoning, $V^p_{ij}$ is continuous.

\end{proof}


Since $V^p_{ji}$ and $V^p_{ij}$ are continuous, we can create a path ${PQ}^{i,j} ~:~ [0,1] \rightarrow XV\left(i,j\right)$ such that all pareto-efficient payoff vectors are in the range of the path as follows.
\begin{align}
{PQ}^{i,j}\left(x\right) = \left(\left(1-x\right) V^\circ_i\left(EX\left(i,j\right)\right) + x V^*_i\left(EX\left(i,j\right)\right), V^p_{ji}\left(\left(1-x\right) V^\circ_i\left(EX\left(i,j\right)\right) + x V^*_i\left(EX\left(i,j\right)\right)\right)\right)
\end{align}
The range ${PQ}^{i,j}([0,1])$ of  ${PQ}^{i,j}$ is the set of all pareto-efficient payoff vectors.

Similarly, since $V^p_{ji}$ maps points in $[V^\circ_i\left(RX\left(i,j\right)\right), V^*_i\left(RX\left(i,j\right)\right)]$ in $[V^\circ_j\left(RX\left(i,j\right)\right), V^*_j\left(RX\left(i,j\right)\right)]$ and $V^p_{ij}$ maps points in $[V^\circ_j\left(RX\left(i,j\right)\right), V^*_j\left(RX\left(i,j\right)\right)]$ in $[V^\circ_i\left(RX\left(i,j\right)\right), V^*_i\left(RX\left(i,j\right)\right)]$, therefore the set of rationally feasible pareto-efficient payoff vectors is a path-connected subset of the set of pareto-efficient payoff vectors. Therefore we can create a path ${RQ}_{i,j} ~:~ [0,1] \rightarrow RV\left(i,j\right)$ such that all rationally feasible pareto-efficient payoff vectors are in the range of the path as follows.
\begin{align}
{RQ}^{i,j}\left(x\right) = \left(\left(1-x\right) V^\circ_i\left(RX\left(i,j\right)\right) + x V^*_i\left(RX\left(i,j\right)\right), V^p_{ji}\left(\left(1-x\right) V^\circ_i\left(RX\left(i,j\right)\right) + x V^*_i\left(RX\left(i,j\right)\right)\right)\right)
\end{align}
The range ${RQ}^{i,j}([0,1])$ of  ${RQ}^{i,j}$ is the set of all rationally feasible pareto-efficient payoff vectors. We will refer to the elements of ${RQ}^{i,j}\left(x\right)$ as ${RQ}^{i,j}_i\left(x\right)$ and ${RQ}^{i,j}_j\left(x\right)$ respectively. Thus  ${RQ}^{i,j}_i$ is strictly increasing and since $V^p_{ji}$ is strictly decreasing, therefore ${RQ}^{i,j}_j$ is strictly decreasing.


\subsection{Exchange network}

Two exchange opportunities $\left(i,j\right)$ and $\left(i,k\right)$ are connected to the degree that exchange in one opportunity is contingent on exchange (or nonexchange) in the other opportunity. (a) The connection is positive if exchange in one opportunity is contingent on exchange in the other, (b) The connection is negative if exchange in one opportunity is contingent on nonexchange in the other \cite{Cook-Emerson:Equity}. In this paper, we restrict our attention to negative connections. Thus each actor can enter into an exchange with at most one of its neighbors.

An {\bf exchange network} is a graph $G=(N,EN)$ where $EN\subseteq E$ is the set of exchange opportunities in the social network along which the exchanges happen. We call these active exchange opportunities in the exchange network $G$. For simplicity, we will refer to $EN$ as the exchange network when $N$ will be understood.


\section{Network Exchange Game}

We now introduce the network exchange game. The network exchange game forms the fundamental structure of analysis in our work on exchange networks. In a negatively connected network, the actors pick at most one neighbor to perform an exchange with and also pick the terms of exchange. An exchange happens if two actors sharing an exchange opportunity pick each other and terms of exchange picked by each is agreeable to the other. An actor obtains a positive payoff only if she performs an exchange.

\label{sec:Network-Exchange-Game}
On the social network $S$, with the capacities of the exchange opportunities $W$, for any actor $i\in N$ with $\left|M_{i}\right|$ units of the item $M_{i}\in M$, let $Nbr\left(i\right)=\left\{ j\in N:\left(i,j\right)\in E\right\} $ be the set of neighbors of actor $i$, determined by the social network $S$.

A {\bf pure strategy} for an actor $i\in N$ is $p_{i}=\left(j,m_{i},m_{j}\right)$ i.e. choose $j\in Nbr\left(i\right)$ to perform an exchange and choose to give $m_{i}\leq\left|M_{i}\right|$ units of $M_{i}$ to $j$ and ask for $m_{j}\leq\left|M_{j}\right|$ of $M_{j}$ from $j$ such that $m_{i}+m_{j}\leq W\left(i,j\right)$. The set of pure strategies for an actor $i$ is $P_{i}\subseteq Nbr(i)\times\mathbf{R}_{+}^{2}$. A {\bf pure strategy profile} is the tuple $p=\left(p_{1},p_{2},...,p_{\left|N\right|}\right)$. The set of pure strategy profiles is $P=\overset{\left|N\right|}{\underset{i=1}{\prod}}P_{i}$. If for any exchange opportunity $(i,j)$, $RX\left(i,j\right) = \phi$ or $RV(i,j) = \{0\}$, then we can remove the exchange opportunity from the network without affecting the payoffs for any outcomes. Therefore, without loss of generality, we assume that for each exchange opportunity $(i,j) \in E$, $RX\left(i,j\right)$ is non-empty and there is at least one positive rationally feasible payoff vector $(v_i,v_j) > 0$ in $RV(i,j)$.

For a pure strategy profile $p\in P$, an exchange opportunity $\left(i,j\right)\in E$ is used for exchange between the actors $i$ and $j$ iff, $p_{i}=\left(j,m_{i,}m_{j}\right)$ and $p_{j}=\left(i,m_{j},m_{i}\right)$, for some $m_{i}\leq\left|M_{i}\right|$, $m_{j}\leq\left|M_{j}\right|$, $m_{i}+m_{j}\leq W\left(i,j\right)$. We then say that the exchange opportunity $\left(i,j\right)\in EN$ is in the exchange network.

For a pure strategy profile $p\in P$, the payoff of the actor $i$ is positive only if $\exists j\in Nbr\left(i\right)$, s.t. and $p_{i}=\left(j,m_{i,}m_{j}\right)$ and $p_{j}=\left(i,m_{j},m_{i}\right)$, for some $m_{i}\leq\left|M_{i}\right|$, $m_{j}\leq\left|M_{j}\right|$, $m_{i}+m_{j}\leq W\left(i,j\right)$; the payoff of $i$ is then $u_{i}\left(p\right)=V_{ij}\left(m_{i},m_{j}\right)$.

Thus a strategy profile $p$ induces a unique exchange network $EN\left(p\right)$, where $EN\left(p\right)=\left\{ \left(i,j\right)\in E:p_{i}=\left(j,m_{i,}m_{j}\right),p_{j}=\left(i,m_{j},m_{i}\right)\right\} $. In the cannonical normal form, the game can be represented by the tuple $\left\langle N,P,U\right\rangle$, where $N$ is the set of selfish actors, $P$ is the set of pure strategy profiles, and $U$ is the payoff function $U:P\rightarrow\mathbf{R}^{\left|N\right|}$.

The social network provides a constraint on the strategies of the actors and the exchange opportunities provide a constraint on the payoffs derived by the actors. An actor's pure strategy determines a neighbor she picks to exchange with and the terms of exchange she proposes. If the proposal is mutually agreed upon, a contract or an exchange relationship is formed. In this one shot game, such exchange relationships form an exchange network. The outcome of the game determines both an exchange network and the payoffs for the individuals. In the following theorem and the corollaries, we identify some properties of the payoffs in an exchange network and some constraints on the structure of the exchange network.

For a given exchange network, $EN$, the set of strategy profiles $P\left(EN\right)=\left\{ p\in P:EN\left(p\right)=EN\right\} $ will be referred to as the supporting strategies for $EN$. The set of payoff profiles supported by $EN$ is $\mathbf{V}\left(EN\right)=\left\{ \mathbf{V}\in\mathbb{R^{\left|N\right|}}:\exists p\in P\left(EN\right)\mbox{ s.t. }U\left(p\right)=\mathbf{V}\right\} $. We will say that $EN$ supports a payoff profile $\mathbf{V}$.


\section{Solution Concept}

The solution concept we explore is of pair-wise stability following \cite{Network-Formation}. We note that Nash equilibrium is not very informative in the exchange network, since formation of any exchange relationship needs actions by two actors, hence, any change in the action of one player cannot lead to formation of a new link and cannot increase her payoff. Hence, a different notion of equilibrium based upon pair wise stability is needed. The following defines the notion of stability in an exchange network.

A strategy profile is pair-wise stable if, no actor can increase her payoff by changing her strategy (changing the neighbor to propose to or proposing different terms of exchange) and no pair of actors having an exchange opportunity described by the graph $S$, can jointly change their strategies (to perform an exchange) such that at least one of the actors increases her payoff and the other actor does not decrease her payoff.

\begin{defn}
Formally, a strategy profile, $p$, is pair-wise stable if $u_{i}\left(p\right)\geq u_{i}\left(p_{i}',p_{-i}\right)$ $\forall i\in N,p_{i}'\in P_{i}$ and $u_{i}\left(p_{i}',p_{j}',p_{-\left\{ i,j\right\} }\right)>u_{i}\left(p\right)\Rightarrow u_{j}\left(p_{i}',p_{j}',p_{-\left\{ i,j\right\} }\right)<u_{j}\left(p\right)$ $\forall\left(i,j\right)\in E,p_{i}'\in P_{i},p_{j}'\in P_{j}$.
An exchange network $EN$ is pair-wise stable if there exists a supporting strategy profile, $p$, for $EN$ such that $p$ is pair-wise stable.
A payoff profile $\mathbf{V}$ is pair-wise stable if there exists a pair-wise stable strategy profile $p$ with $U\left(p\right)=\mathbf{V}$.
\end{defn}

We now provide an equivalent characterization of the pair-wise stable strategy profile that will be used to show the existence of a pair-wise stable strategy profile.
\begin{Lemma} \label{lem:pareto-dominated}

If a strategy profile $p^*$ in the network exchange game is not pair-wise stable then at least one of the conditions holds:

\begin{enumerate}

\item there exists and actor $i$ with strategy $p^*_i = (j, m_i, m_j)$, such that $V_{ij}(m_i,m_j) < 0$.

\item there exists an exchanges opportunity, $(i,j) \in E$, and rationally feasible pareto-efficient payoff vector $(v_i, v_j) \in {RQ}^{i,j}([0,1])$, such that $v_i \geq u_i(p^*)$ and $v_j \geq u_j(p^*)$ with at least one strict inequality.

\end{enumerate}

\end{Lemma}

\begin{proof}

Assume $p^*$ is not pair-wise stable and that each actor $i \in N$, has a strategy $p^*_i = (j, m_i, m_j)$, such that $V_{ij}(m_i,m_j) \geq 0$. First we note that if an actor $i$ can unilaterally deviate to a strategy $p^{'}_i = (j, m_i, m_j)$ and increase her payoff then $p^*_j = (i,m_i,m_j)$. Therefore, $u_j(p^*) = 0$. Since $RV(i,j)$  has at least one non-zero rationally feasible pareto-efficient payoff vector therefore there exists a rationally feasible pareto-efficient payoff vector $(v_i, v_j) = {RQ}^{i,j}(1)$ such that $v_i = V^*_{i}(RX(i,j)) \geq V_{ij}(m_i,m_j) = u_i(p^{'}_i, p^*_{-1})$ and $v_j = 0$. Therefore the second condition is satisfied.

Secondly, if two actors $i,j$ with $(i,j) \in E$ can deviate simultaneously their strategies to $p^{'}_i = (j, m_i, m_j)$ and $p^{'}_j = (j, m_i, m_j)$ to increase their payoffs. Without loss of generality assume $V_{ij}(m_i,m_j) > u_i(p^*)$ and $V_{ji}(m_i,m_j) \geq u_j$. Then $V^p_{ji}(V_{ij}(m_i,m_j)) \geq V_{ji}(m_i,m_j) \geq u_j(p^*)$ and $(V_{ij}(m_i,m_j), V^p_{ji}(V_{ij}(m_i,m_j))) \in {RQ}^{i,j}([0,1])$.


\end{proof}

The characterization of pair-wise stable strategy profile in lemma \ref{lem:pareto-dominated} will be used in the next section to prove the existence of a pair-wise stable strategy profile.

\section{Existence of Pair-wise Stable Strategy Profiles and Pair-wise Stable Exchange Networks}

To show the existence of a pair-wise stable strategy profile, we reduce the problem to a similar problem in \cite{Generalized-Stable-Matching}. In \cite{Generalized-Stable-Matching}, we showed the existence of a weighted stable matching. We will show that there is a corresponding pair-wise stable strategy profile. We now state the existence theorem.

\begin{Theorem}

For any network exchange game $\left\langle N,P,U\right\rangle$ as introduced in section \ref{sec:Network-Exchange-Game}, there exists a pair-wise-stable strategy profile $p^*$.

\end{Theorem}

\begin{proof}

Arbitrarily pick a network exchanges game $\left\langle N,P,U\right\rangle$ as introduced in section \ref{sec:Network-Exchange-Game} and consider the following system of inequalities:
\begin{align}
&s_i + s_j = 1, \forall \left(i,j\right) \in EN\\
&u_i = RQ^{i,j}_i\left(s_i\right), u_j = RQ^{i,j}_j\left(1-s_j\right), \forall \left(i,j\right) \in EN\\
&\left(RQ^{i,j}_i\right)^{-1}\left(u_i\right) + 1 - \left(RQ^{i,j}_j\right)^{-1}\left(u_j\right) \geq 1, \forall \left(i,j\right) \in E\\
&s_{i} \geq 0, \forall i \in N
\end{align}
In \cite{Generalized-Stable-Matching}, we showed the existence of feasible solutions to the above system of inequalities when $RQ^{i,j}_i$ is continuous and strictly increasing and $RQ^{i,j}_j$ continuous and is strictly decreasing. Pick a feasible solution $(\mathbf{s}^*, EN^*)$ for the system and create a strategy profile for the network exchange game as follows:

\begin{itemize}

\item For each $(i,j) \in EN^*$, set $p^*_i = (j, m^{p_i}_i(RQ^{i,j}_i\left(s_i\right)), m^{p_i}_j(RQ^{i,j}_i\left(s_i\right)))$ and
$$p^*_j = (i, m^{p_i}_i(RQ^{i,j}_i\left(s_i\right)), m^{p_i}_j(RQ^{i,j}_i\left(s_i\right)))$$
where $(m^{p_i}_i(RQ^{i,j}_i\left(s_i\right)), m^{p_i}_j(RQ^{i,j}_i\left(s_i\right)))$ is the unique exchange satisfying \ref{eqn:char-pareto-efficient}.

\item For each unmatched $i \in N$, arbitrarily pick a neighbor $j \in Nbr(i)$ and set $p^*_i = (j, 0, 0)$.

\end{itemize}

We now show that $p^*$ is pair-wise stable. Assume $p^*$ is not pair-wise stable. First we note that for all $i \in N$, the strategy $p^*_i = (j, m_i, m_j)$ is such that $V_{ij}(m_i,m_j) \geq 0$. Then by lemma \ref{lem:pareto-dominated} there exits an exchange opportunity $(i,j) \in E$ and a rationally feasible pareto-efficient payoff vector $(v_i, v_j) \in RQ^{i,j}([0,1])$ such that $v_i \geq u_i(p^*)$ and $v_j \geq u_j(p^*)$ with at least one strict inequality. Pick any such exchange opportunity $(i,j)$. Clearly, $(i,j) \notin EN$ because by definition $(u_i(p^*), u_j(p^*))$ is a pareto-efficient payoff vector for the exchange opportunity $(i,j)$. Without loss of generality, assume $v_i \geq u_i(p^*)$. Define $\lambda_1 = \left(RQ^{i,j}_i\right)^{-1}\left(u_i(p^*)\right)$ and $\lambda_2 = \left(RQ^{i,j}_i\right)^{-1}\left(v_i\right)$. Then $\lambda_2 > \lambda_1$ and from the definition 12, $u_j(p^*) \geq RQ^{i,j}_j(\lambda_1)$. Therefore $u_j(p^*) \geq RQ^{i,j}_j(\lambda_1) > RQ^{i,j}_j(\lambda_2) = v_j$ which contradicts the assumption that $v_j \geq u_j(p^*)$. Since $(i,j)$ was arbitrarily picked, therefore there does not exist an exchange opportunity $(i,j) \in E$ and a rationally feasible pareto-efficient payoff vector $(v_i, v_j) \in RQ^{i,j}([0,1])$ such that $v_i \geq u_i(p^*)$ and $v_j \geq u_j(p^*)$ with at least one strict inequality. Therefore by lemma \ref{lem:pareto-dominated}, $p^*$ is pair-wise stable.

\end{proof}

Thus we have proved the existence of a pair-wise stable strategy profile in every network exchange game as introduced in section \ref{sec:Network-Exchange-Game}. The pair-wise stable strategy profile induces a pair-wise stable exchange network and a pair-wise stable payoff profile and hence, every network exchange game as introduced in section \ref{sec:Network-Exchange-Game} has a pair-wise stable stable exchange network and a pair-wise stable payoff profile.


\section{Conclusion}

In this paper we show that when the individuals in a bipartite network exclusively choose partners and exchange valued goods with their partners, then there exists a set of exchanges that are pair-wise stable.


\bibliographystyle{plain}
\bibliography{Ankur-Mani-bibtex}

\end{document}